\documentclass[12pt]{article}

\usepackage{amsmath, amsthm, amssymb}
\usepackage{graphicx}
\usepackage{fullpage}
\usepackage{xcolor}
\usepackage{natbib}
\usepackage{array,booktabs,makecell,multirow}
\usepackage[utf8]{inputenc}
\usepackage[T1]{fontenc}
\usepackage{float}
\usepackage{multirow}
\usepackage{indentfirst}
\usepackage{csquotes}
\usepackage[english]{babel}	
\usepackage{subfig}
\usepackage[skip = 13pt]{caption}
\usepackage[nottoc,numbib]{tocbibind}
\RequirePackage[colorlinks,citecolor=blue,urlcolor=blue]{hyperref}
\usepackage{xcolor}

\newtheorem{lemma}{Lemma}
\newtheorem{cor}{Corollary}
\newtheorem{theorem}{Theorem}

\DeclareMathOperator{\diag}{diag}

\DeclareMathOperator{\Var}{Var}
\DeclareMathOperator{\Tr}{Tr}
\DeclareMathOperator{\const}{const.}

\DeclareMathOperator{\GCV}{GCV}
\DeclareMathOperator{\median}{median}

\title{M-type penalized splines with auxiliary scale estimation}
\author{Ioannis Kalogridis and Stefan Van Aelst}
\begin{document}
\maketitle

\begin{abstract}
Penalized spline regression is a popular and flexible method of obtaining estimates in nonparametric models but the classical least-squares criterion is highly susceptible to model deviations and atypical observations. Penalized spline estimation with a resistant loss function is a natural remedy, yet to this day the asymptotic properties of M-type penalized spline estimators have not been studied. We show in this paper that M-type penalized spline estimators achieve the same rates of convergence as their least-squares counterparts, even with auxiliary scale estimation. We further find theoretical justification for the use of  a small number of knots relative to the sample size. We illustrate the benefits of M-type penalized splines in a Monte-Carlo study and two real-data examples, which contain atypical observations.
\end{abstract}

\section{Introduction}

Based on data $(x_i, Y_i), \ldots, (x_n, Y_n)$ with fixed $x_i \in [a,b]$ and $-\infty<a<b < \infty$, consider the classical nonparametric regression model
\begin{equation}
\label{eq:1}
Y_i = f(x_i) + \epsilon_i,
\end{equation}
where $f(\cdot)$ is a sufficiently smooth function which we shall endeavour to estimate and the $\epsilon_i, \ i=1, \ldots, n$ are independent and identically distributed error terms, commonly assumed to have zero mean and finite variance $\sigma^2$.

Nonparametric regression has been a burgeoning field for many years and many ingenious methods have been proposed for the estimation of the regression function $f(\cdot)$. These methods broadly comprise kernel regression, orthogonal series, splines and wavelets, see, e.g., \citet{Was:2006} for an overview. In this paper, we focus on robust estimation with penalized splines. Owing to their ease of fitting and flexible choice of knots and penalties, penalized splines have been exceedingly popular in recent years following the seminal works of \citet{O:1986} and \citet{Eilers:1996}. Penalized splines offer a compromise between the simplicity of (unpenalized) regression splines and the computational complexity of smoothing splines, see \citet[Chapter 7]{Wahba:1990} for this point. Asymptotic properties of least-squares penalized spline estimators have been studied by \citet{Hall:2005}, who established their consistency, \citet{Li:2008}, who derived the equivalent kernel representation for lower-order splines, and more broadly by \citet{Claeskens:2009}, who obtained rates of convergence for a variety of settings.

It is well-known that estimators derived from the minimization of an $L_2$ norm are susceptible to atypical observations. That is, a single gross outlier can significantly distort the estimates as well as inferences based on them. This fact has motivated proposals that aim to achieve some degree of resistance towards outlying observations. \citet{Lee:2007} proposed replacing the squared loss with a loss function that increases more slowly and supplied an algorithm based on pseudo-observations. In the same vein, \citet{T:2010} proposed minimizing a robust scale of the residuals with an additional penalty term in order to obtain resistant estimates. However, despite the well-studied theoretical properties of least-squares penalized splines, it is curious that no asymptotic results have been established outside that framework. This comes in stark contrast to robust smoothing and regression splines whose asymptotic properties have been established as early as \citet{Cox:1983} and \citet{Shi:1995} with significant extensions with respect to scale estimation offered by \citet{Cunningham:1991} and \citet{He:1995} respectively.

To fill this gap we study the consistency and establish rates of convergence of general M-type penalized spline estimators with a derivative-based penalty. Moreover, we also consider the case where a preliminary scale estimator is used to standardize the residuals in the loss function. As mentioned previously, this type of estimator was already considered by \citet{Lee:2007} from a computational point of view, but without any theoretical support. Our approach further differs from their proposal in three key aspects. Firstly, we use the nearly orthogonal B-spline basis instead of the badly conditioned truncated polynomials and derive theoretical properties based on that representation. Secondly, we use a robust preliminary scale estimate that does not require model fitting as opposed to the originally proposed concomitant scale estimate. This dramatically reduces the computational burden as there is no need anymore to iterate between coefficient and scale estimates. Finally, we propose a fast, effective and automatic method of selecting the penalty parameter which has the advantage that the penalty parameter no longer needs to be chosen at each iteration of the algorithm.

The rest of the paper is structured as follows. We review a few basic facts about spline estimation in Section 2 and describe the proposed M-type penalized estimator and its computation in Section 3. In Section 4 the asymptotic properties of the estimator are studied. We obtain an asymptotic expansion of the estimator and show that it achieves the optimal rates of convergence, even with auxiliary scale estimation. Section 5 examines the performance of the proposed method via a Monte-Carlo study while Section 6 illustrates the advantages of its use in two real data examples. Finally, some possible directions for future research are discussed in the concluding section of this paper.

\section{Spline-based estimation}

A spline is defined as a piecewise polynomial that is smoothly connected at its joints (knots). More specifically, for any fixed integer $p \geq 1$, denote $S_{K}^p$ the set of spline functions of order $p$ with knots $a=t_0<t_1 \ldots <t_{K+1}=b$. Then for $p=1,\ S_{K}^1$ is the set of step functions with jumps at the knots and for $p\geq 2$,

\begin{equation*}
S_{K}^p = \left\{ s\in \mathcal{C}^{p-2}[a,b]: s(x) \ \text{is a polynomial of degree $(p-1)$ on each subinterval $[t_i, t_{i+1} ]$} \right\}.
\end{equation*}
It is easy to see that $S_{K}^p$ is a $p+K$ dimensional subspace of $\mathcal{C}^{p-2}[a,b]$ and a basis may be derived by means of the truncated polynomials $1, x, \ldots, x^{p-1}, (x-t_1)_{+}^{p-1}, \ldots (x-t_{K})_{+}^{p-1}$ with $(x)_{+} = x \mathcal{I}(x\geq 0)$. However, truncated polynomial functions are known to be highly collinear for a large number of knots. Hence, it is preferable to use the more stable B-spline basis for $S_{K}^p$, which we now briefly discuss; we refer to \citet[Chapter IX]{DB:2001} for a full treatment.

The B-spline functions may be defined recursively but they may also be directly derived as linear combinations of the truncated polynomial functions $(x-t_1)_{+}^{p-1}, \ldots, (x-t_{K})_{+}^{p-1} $. Several important properties follow immediately from this construction. Let $\{t_i\}_{i=1}^{K+2p}$ be an augmented and relabelled sequence of knots obtained by repeating $t_0$ and $t_{K+1}$ exactly $p$ times. The B-spline basis for the family $S_{K}^p$ is given by

\begin{equation}
\label{eq:2}
B_{K,i}(x) = \left(t_{i+p}-t_i\right) \left[t_i, \ldots, t_{i+p} \right] (t-x)_{+}^{p-1}, \quad i=1, \ldots, K+p
\end{equation}
where for a function $g$ the placeholder notation $\left[t_i, \ldots, t_{i+p} \right]g$ denotes the $p$th order divided difference of $g(\cdot)$ at $t_i, \ldots, t_{i+p}$. Among other interesting properties B-splines of order $p$ satisfy

\begin{itemize}
\item[(a)] Each $B_{K,i}$ is a polynomial of order $p$ on each interval $(t_i, t_{i+1})$ and has $(p-2)$ continuous derivatives.
\item[(b)] $0<B_{K,i }(x) \leq 1$ for $x \in (t_i, t_{i+p})$ and $B_{K,i}(x)=0$ otherwise. 
\item[(c)] $\sum_{i=1}^{K+p} B_{K,i}(x) = 1$ for all $x \in [a,b]$.
\end{itemize}
Property (b) is referred to as the local support property of the B-spline basis and it is one of the reasons for the popularity of the basis in digital computing and functional approximation. Further properties of splines and the B-spline basis may be found in the classical monographs of \citet{Devore:1993}, \citet{DB:2001}  and \citet{Schumaker:2007}. Spline functions from a statistical perspective are covered in, e.g.,  \citet{Wahba:1990}, \citet{Green:1993}, \citet{Eubank:1999} and \citet{Gu:2013}.

Of particular interest are the approximation properties of spline functions. Provided that $f(\cdot)$ is a sufficiently smooth function, in the sense  of having a number of continuous derivatives, the spline approximation theorems, see, e.g., \citet[Chapter 6]{Schumaker:2007}, allow us to deduce that $f(\cdot)$ may be well-approximated by a spline function $f^{\star}(x) = \sum_{j=1}^{K+p} \beta_j^{\star} B_{K,j}(x)$. This fact underlies all spline-based estimation techniques except for the smoothing spline construction, see the end of this section. A reasonable approximation may thus be constructed by expanding $f(\cdot)$ in the B-spline basis and estimating the coefficient vector $\boldsymbol{\beta}$. The most popular estimation method to that end is the least squares criterion leading to the minimization of

\begin{equation}
\label{eq:3}
n^{-1}\sum_{i=1}^n \left\{ Y_i - \sum_{j \leq K+p} \beta_j B_{K, j}(x_i) \right\}^2.
\end{equation}
See \citet{Ag:1980}, \citet{Wegman:1983} and \citet{Shen:1998} for more details on least-squares regression splines. To compensate for the lack of robustness of the least-squares criterion \citet{Shi:1995}  proposed minimizing instead
\begin{equation}
\label{eq:4}
n^{-1}\sum_{i=1}^n  \rho \left( Y_i - \sum_{j \leq K+p} \beta_j B_{K, j}(x_i) \right),
\end{equation}
for a suitably chosen convex function $\rho$, see \citet{Huber:1973}. Unfortunately, a common drawback in both procedures is the sensitivity of the estimator to the number of knots as well as their position. Selection procedures are non-trivial and very often involve either a stepwise/backward knot placement procedure or minimization of a complex information criterion, \citet[Chapter 6]{Eubank:1999}. 

An alternative estimation method results from adding a ridge-type roughness penalty to the above minimization problem effectively shifting the focus from the knots to the penalty parameter. In particular, \citet{O:1986} proposed adding a roughness penalty in the form of the integrated squared $q$th derivative of the spline function and minimize

\begin{equation}
\label{eq:5}
n^{-1}\sum_{i=1}^n \left\{ Y_i - \sum_{j \leq K+p} \beta_j B_{K, j}(x_i) \right\}^2 + \lambda \int_{[a,b]} \left[ \left\{ \sum_{j \leq K+p} \beta_j B_{K, j}(t) \right\}^{(q)} \right]^2 dt ,
\end{equation}
for some knot sequence $t_1, \ldots, t_{K+2p}$. In the original proposal, $q$ was equal to 2 but generalization to higher order penalties is straightforward. The present penalized spline estimator, commonly referred to as O-spline, has two smoothing parameters: the number of knots and the penalty parameter. The inclusion of the penalty parameter, however, affords us the use of a large number of knots. It is customary to select these knots in quasi-automated manner, for example, a large number of them may be placed at the quantiles of the $x_i$. The penalty parameter $\lambda$ is usually chosen either through cross-validation methods or through the mixed-model connection. We refer to \citet{Ruppert:2003, Wood:2017} for more details and illustrative examples for the case of the least-squares penalized spline estimator.

Finally, we briefly review smoothing spline estimators. These estimators arise as solutions to the variational problem
\begin{equation}
\label{eq:6}
\min_{f \in \mathcal{W}^{q, 2}([a,b])} \left[ n^{-1} \sum_{i=1}^n \{Y_i - f(x_i)\}^2 + \lambda \int_{[a, b]} \left\{f^{(q)}(x)\right\}^2 dx \right],
\end{equation}
where $\mathcal{W}^{q, 2}([a,b])$ refers to the Sobolev space of order $q$, that is,
\begin{align*}
\mathcal{W}^{q, 2}([a,b]) = \{ f:[a,b] \to \mathbb{R}, f\ &\text{has $q-1$ absolutely continuous derivatives}  \\  & f^{(1)}, \ldots, f^{(q-1)}\ \text{and} \int_a^b \{f^{(q)}(x)\}^2 dx< \infty  \},
\end{align*}
See \citet{Adams:2003} for more details on Sobolev spaces. It is a remarkable fact that without any constraints the solution to the above problem is a special kind of spline: a natural polynomial spline of degree $2q-1$ with knots at $x_i$. This spline may also be written as a linear combination of B-spline functions with special care so that the boundary conditions are respected, see \citet{DB:2001}. Since this spline is the unique solution to the problem, smoothing splines, unlike regression and penalized splines, incur no approximation error. More details on least-squares smoothing splines may be found in \citet{Eubank:1999} while M-type smoothing splines are discussed in \citet{Huber:1979}, \citet{Cox:1983}, \citet{Cunningham:1991} and \citet{Eg:2009}, who study the $L_1$ smoothing spline in detail.

\section{M-type penalized spline estimators}

\subsection{Estimating equations and preliminary scale estimation}
Up to now M-type penalized spline estimators have received much less attention in the literature in comparison to both M-type regression splines and smoothing splines. However,
penalized splines are situated in between regression and smoothing splines and as such are well-suited for a wide variety of problems. To be precise, the estimator $\widehat{f}(\cdot) = \sum_{j} B_{K,j}(\cdot) \widehat{\beta}_j$ that we consider minimizes

\begin{equation}
\label{eq:7}
\frac{1}{n}\sum_{i=1}^n \rho \left( Y_i - \sum_{j \leq K+p} \beta_j B_{K, j}(x_i) \right) + \lambda \int_{[a,b]} \left[ \left\{ \sum_{j \leq K+p} \beta_j B_{K, j}(t) \right\}^{(q)} \right]^2 dt ,
\end{equation}
for some nonnegative $\rho$ that is symmetric about zero, satisfies $\rho(0)=0$ and can be either convex or non-convex. The choice $\rho(x) = x^2$ brings us back to the least squares minimization problem, but \eqref{eq:7} allows for more general functions that reduce the effect of large residuals. Examples include Huber's function \citep{Huber:1964} given by
\begin{equation}
\rho_c(x) = \begin{cases} (1/2)x^2, & |x| \leq c \\ c|x|-(1/2)c^2, & |x| >c
\end{cases}
\end{equation}
where the constant $c$ regulates the degree of resistance. For large values of $c$ one essentially obtains the ordinary quadratic loss function but for smaller values a higher degree of robustness is achieved.  Many other $\rho$-functions can be constructed by imitating parametric likelihood models, such as the Cauchy, logistic and Laplace models. The case of a Laplace model, in particular, leading to the $L_1$ loss function may be understood as a limiting case of the Huber loss function for $c \to 0 $.

For ease of notation define the spline basis vector evaluated at $x$ as $\mathbf{B}(x)$, that is, $\mathbf{B}(x) := \{B_{K, 1}(x), \ldots, B_{K,K+p}(x) \}^{\top}$, denote the $n \times (K+p)$ spline design matrix with $\mathbf{B} := \{ \mathbf{B}^{\top}(x_1), \ldots, \mathbf{B}^{\top}(x_n) \}^{\top}$ and let $\mathbf{D}_{i,j} = \int B_{K, j}^{(q)} B_{K, i}^{(q)}$. See \citet[pp. 116-117]{DB:2001} for derivative expressions for B-splines. With this notation, it is easy to see that the minimizer $\boldsymbol{\widehat{\beta}}$ satisfies

\begin{equation*}
-\frac{1}{n} \sum_{i=1}^n \rho^{\prime} \left(Y_i - \mathbf{B}^{\top}(x_i) \boldsymbol{\widehat{\beta}} \right) \mathbf{B}(x_i) + 2\lambda \mathbf{D} \boldsymbol{\widehat{\beta}} = \mathbf{0}.
\end{equation*}
The solution is unique for strictly convex $\rho$-functions but non-unique otherwise.

To include a preliminary scale estimate $\widehat{\sigma}$ it suffices to modify the $\rho$ function according to $\rho_{\widehat{\sigma}}(x) := \rho(x/\widehat{\sigma})$. Traditionally in robust statistics, see, e.g., \citet{Maronna:2006}, $\widehat{\sigma}$ is obtained by an initial robust regression fit to the data. This may be avoided by using the technique of pseudo-residuals as in \citet{Cunningham:1991}. Specifically, assuming $x_1<\ldots < x_n$, let 

\begin{align}
\widehat{\epsilon}_i = w_i Y_{i-1} + s_i Y_{i+1} - Y_i, \quad i= 2, \ldots, n-1
\end{align}
with
\begin{align*}
w_i  = (x_{i+1}-x_i)/(x_{i+1} - x_{i-1}) \quad \text{and} \quad s_i  = (x_i-x_{i-1})/(x_{i+1}-x_{i-1}).
\end{align*}
It is easy to see that the pseudo-residuals are constructed by using straight line fits on two outer observations in order to predict the middle observation. \citet{Gasser:1986} proposed estimating $\sigma^2$ in \eqref{eq:1} using 

\begin{equation*}
\widehat{\sigma}^2 = \frac{1}{n-2} \sum_{i=2}^{n-1} \frac{ \widehat{\epsilon}_i^2}{w_i^2+s_i^2+1},
\end{equation*}
where the standardization results from noticing that $\mathbb{E}\{|\widehat{\epsilon}_i|^2 \} = (w_i^2+s_i^2+1)\sigma^2 + O(n^{-2})$ for $f(\cdot)\in \mathcal{C}^2[a,b]$. The sample variance is not robust with respect to outliers but robust estimates can also be obtained from pseudo-residuals. For example, one may compute the median absolute deviation, an M-scale or the inter-quartile range (IQR), as suggested by \citet{Cunningham:1991} in the context of smoothing splines. Another class of robust scale estimators may be constructed using pairwise differences of $Y_{i}$. In particular, \cite{Boente:1997} propose estimating $\sigma$ with

\begin{equation*}
\widehat{\sigma} = (2^{1/2} 0.6745)^{-1} \median|Y_{i+1} - Y_{i}|,
\end{equation*}
which may be viewed as a robust alternative to the well-known Rice estimator.

It should be noted that contrary to unpenalized regression, such as regression splines, standardizing with a scale estimate will not, in general, lead to scale equivariant  estimates. This will be approximately the case, though, provided that the penalty term is negligible, that is, either $\lambda$ is small or $\int \{\widehat{f}^{(q)}\}^2$ is small, i.e., the estimating function is "close" to being a polynomial. Nevertheless, the inclusion of a robust scale estimate leads to useful diagnostic tools for outliers, see the real data examples of Section 6 for some interesting illustrations.

\subsection{Computation and smoothing parameter selection}

The success of any penalized spline estimator, least-squares and robust alike, rests on appropriate selection of the smoothing parameters: the dimension of the spline basis and the penalty parameter. Here, we shall assume that the order of the spline and the penalty has been fixed in advance by the practitioner, common choices being $p=4$ and $q=2$. First, we make a brief note on the computation of the penalized estimates.

The solution to \eqref{eq:7} may be computed efficiently by a modification of the well-known iteratively reweighted least squares (IRWLS) algorithm \citep{Maronna:2006}. Define $r_i(\boldsymbol{\beta}) = Y_i - \mathbf{B}^{\top} (x_i)\boldsymbol{\beta}$, let $W_{i}(\boldsymbol{\beta}) = \rho^{\prime}(r_i(\boldsymbol{\beta})/\widehat{\sigma})/(r_i(\boldsymbol{\beta})/\widehat{\sigma}), \ i = 1, \ldots, n$ and put $\mathbf{W}(\boldsymbol{\beta}) = \diag\{ W_i(\boldsymbol{\beta})\}$. It can be seen that $\widehat{\boldsymbol{\beta}}$ satisfies

\begin{equation}
- (n \widehat{\sigma}^2)^{-1}  \sum_{i=1}^n W_i(\widehat{\boldsymbol{\beta}})\left\{Y_i - \mathbf{B}^{\top}(x_i) \widehat{\boldsymbol{\beta}} \right\} \mathbf{B}(x_i) + 2 \lambda \mathbf{D} \widehat{\boldsymbol{\beta}} = \mathbf{0}.
\end{equation}
Thus $\widehat{\boldsymbol{\beta}}$ is the minimizer of a weighted penalized least-squares criterion. This suggests an iterative scheme for the computation of $\widehat{\boldsymbol{\beta}}$. At the $m$th step one defines weights $W_{i}(\boldsymbol{\beta}^{(m)}), \ i = 1, \ldots, n$ and obtains the updated approximation to $\widehat{\boldsymbol{\beta}}$, $\boldsymbol{\beta}^{(m+1)}$, by solving

\begin{equation*}
\left\{n^{-1}\mathbf{B}^{\top} \mathbf{W}(\boldsymbol{\beta}^{(m)})\mathbf{B} + 2 \widehat{\sigma}^2 \lambda \mathbf{D}  \right\} \boldsymbol{\beta}^{(m+1)} = n^{-1}\mathbf{B}^{\top} \mathbf{W}(\boldsymbol{\beta}^{(m)}) \mathbf{Y}.
\end{equation*}
It follows from arguments given in \citet{Maronna:2006, Huber:2009} that the procedure is guaranteed to converge to $\boldsymbol{\widehat{\beta}}$, independently of the starting point, provided that $\rho$ is convex, symmetric about zero and $\rho^{\prime}(x)/x$ is bounded and monotone decreasing for $x>0$. Omitting the convexity assumption has the consequence that the algorithm still converges but convergence may instead be to a local minimum. Thus, the choice of the starting point becomes crucial.

Due to the banded structure of the matrices involved, 	successive systems of linear equations can be solved in $O(K)$ computations, after forming the necessary matrices. This needs to be contrasted with the $O(n)$ computations that would have been required by smoothing splines with the B-spline basis. Since often $K<<n$, penalized spline estimators require much less computational effort, particularly for large datasets. Moreover, as \citet[Chapter 7]{Wahba:1990} notes, a ridge regression type argument shows that there always exists a $\lambda>0$ such that the mean-squared error of the corresponding penalized spline estimator is smaller than for $\lambda=0$, i.e., for the regression spline estimator. These facts illustrate the balance penalized splines seek to achieve.

To implement the estimator we follow the recommendation in \citet[Chapter 5]{Ruppert:2003} for the number and location of knots.  Specifically, we take 

\begin{equation}
K = \min\{1/4\ \times\ \text{number of unique}\ x_i, \ 40 \},
\end{equation}
and for the interior knots
\begin{equation}
t_{k} = \left( \frac{k+1}{K+2} \right)\text{th sample quantile of the unique $x_i$},
\end{equation}
for $k = 1, \ldots, K$. Both $K$ and the location of the knots are chosen independently of $\lambda$, as experience with penalized splines has shown that $\lambda$ is more important than $K$, provided that the latter quantity is taken large enough. To choose $\lambda$ we use the generalized cross-validation (GCV) criterion adapted from \citet{Cunningham:1991}, that is,

\begin{equation}
\GCV(\lambda) = \frac{n^{-1} \sum_{i=1}^n W_i(\widehat{\boldsymbol{\beta}}) \left\{  Y_i - \mathbf{B}^{\top}(x_i) \widehat{\boldsymbol{\beta}} \right\}^2 }{ \left\{1- n^{-1}\Tr\mathbf{H}(\lambda) \right\}^2  }
\end{equation}
where $\mathbf{H}(\lambda) = \mathbf{B} \{\mathbf{B}^{\top}\mathbf{W}(\boldsymbol{\widehat{\beta}}) \mathbf{B} + 2 n \lambda \widehat{\sigma}^2 \mathbf{D} \}^{-1}  \mathbf{B}^{\top} \mathbf{W}(\boldsymbol{\widehat{\beta}}) \mathbf{Y} $ is the hat matrix obtained upon convergence of the algorithm. We choose $\lambda$ as the minimizer of this function.

The minimization is usually carried out with a blind grid search leaving to the user the awkward specification of the candidate penalty values. In order to produce a fully automatic method we recommend using a numerical derivative-free optimizer such as the Nelder-Mead method \citep[238--240]{Nocedal:2006}. The method is available in standard software, converges fast and, in our experience, works well for a wide variety of problems. It is therefore our preferred choice for the simulation experiments and the real data analyses presented herein. 

\section{Asymptotic properties}

We now investigate the rates of convergence of the M-type penalized spline estimator defined in \eqref{eq:7}, both with and without the use of an auxiliary scale estimate. For the purpose of comparison, we first list the asymptotic mean-squared errors of regression and smoothing spline estimators under their respective assumptions. For either least-squares or M-type  regression spline estimates defined in \eqref{eq:3}-\eqref{eq:4}, denoted generically by $\widehat{f}_{\text{rsp}}(\cdot)$, one has

\begin{equation}
\label{eq:14}
\frac{1}{n} \sum_{i=1}^n \{\widehat{f}_{\text{rsp}}(x_i) - f(x_i) \}^2 = O_{P}\left(\frac{K}{n}\right) + O_{P}\left(K^{-2p}\right),
\end{equation}
for $f(\cdot)\in \mathcal{C}^{p}[a,b]$, see \citet{Shi:1995}. On the other hand, for least-squares and M-type smoothing splines, denoted generically by $\widehat{f}_{\text{smsp}}(\cdot)$, one has
\begin{equation}
\label{eq:15}
\frac{1}{n} \sum_{i=1}^n \{\widehat{f}_{\text{smsp}}(x_i) - f(x_i) \}^2 = O_{P}\left(\frac{1}{n \lambda^{1/2q}}\right) + O_{P}\left(\lambda\right),
\end{equation}
for $f(\cdot) \in \mathcal{W}^{q, 2}([a,b])$, as seen from the results of \citet{Wahba:1990} and \citet{Cox:1983}. It follows from these results that with appropriate selection of the smoothing parameters $K$ and $\lambda$, $\widehat{f}_{\text{rsp}}(\cdot)$ and $\widehat{f}_{\text{smsp}}(\cdot)$ can attain the optimal rates of convergence for $\mathcal{C}^{p}([a,b])$ and $\mathcal{W}^{q, 2}([a,b])$ functions respectively \citep{Stone:1982}. Since $\mathcal{C}^{q}([a,b]) \subset \mathcal{W}^{q, 2}([a,b])$ smoothing splines require somewhat milder smoothness assumptions in order to attain the same rates of convergence. 

It should be noted that the above results cannot be directly extended to M-type penalized splines since regression splines do not take into account the effect of the penalization while smoothing splines ignore the approximation error incurred by the sieved nature of penalized splines. An independent treatment is thus required. The assumptions that will be needed for our theoretical development are as follows.

\begin{itemize}
\item[A.1]For the unique knots $\{t_i\}_{i = p}^{K+p}$ define $h_i := t_i-t_{i-1}$ and $h := \max_i h_i$. Assume that $\max_i|h_{i+1}-h_i| = o(K^{-1})$ and there exists a constant $M$ such that $(h/\min_i h_i) \leq M$.
\item[A.2] For deterministic design points $x_i \in [a,b], \ i =1, \ldots, n$, assume that there exists a distribution function $Q$ with corresponding continuous density $w$ bounded away from zero and infinity such that, with $Q_{n}$ the empirical distribution of $x_1, \ldots, x_n, \ \sup_{x} |Q_n(x)-Q(x)| = o(K^{-1})$. 
\item[A.3] The number of knots $K=o(n)$.
\end{itemize}
Assumption 1 essentially requires that the knots are quasi-uniform and dense in $[a,b]$.  Assumption 2 is a weak restriction on the knot distribution and finally, assumption 3 puts a limit to the rate of growth of the knots, that is, the number of predictor variables in the regression model. This is a common assumption for sieved estimators, see \citet{Eubank:1999} and \citet{Eg:2009}. All three assumptions are also used for the least-squares setting. For the M-type estimators considered herein we require the following additional assumption on the $\rho$-function and the distribution of $\epsilon$. 

\begin{itemize}
\item[A.4] For $\psi(\cdot) := \rho^{\prime}(\cdot)$ we require $\psi(\cdot) \in \mathcal{C}^2(-\infty, \infty)$ and satisfies $\sup_{x} |\psi^{\prime \prime}(x)| < \infty$, $\mathbb{E} \{\psi(\epsilon)\} = 0$, $\mathbb{E}\{ \psi^{\prime}(\epsilon)\} > 0$, $\mathbb{E} \{ |\psi(\epsilon)|^2 \} <\infty$ and $ \mathbb{E}\{ |\psi^{\prime}(\epsilon)|^2\} < \infty$.
\end{itemize}
Examples of $\rho$-functions that satisfy the smoothness conditions are the convex logistic and the non-convex Tukey bisquare. The Huber $\rho$-function does not meet these requirements but smoothed, yet asymptotically equivalent versions of it do, see \citet{Hampel:2011a} for a possible smoothing scheme. As \citet{Huber:1973} notes, the smoothness conditions on $\rho$ are technically convenient but seem hardly essential for the results to hold. The moment conditions involving the $\psi$-function occur very often in the context of robust estimation, see \citet{Maronna:2006}. In essence, they are  identifiability (Fisher-consistency) conditions so that the correct function $f(\cdot)$ is estimated.

Following \citet{Huber:1973} and \citet{Cox:1983} we aim to approximate the M-type penalized spline estimator with a sequence of special least-squares estimators. To that end, let us define the pseudo-observations

\begin{equation}
\label{eq:16}
\widetilde{Y}_i = f(x_i) + \frac{\psi(\epsilon_i)}{\mathbb{E} \{\psi^{\prime}(\epsilon)\} },
\end{equation}
and let $\widetilde{f}(\cdot) := \mathbf{B}^{\top}(\cdot) \widetilde{\boldsymbol{\beta}}$ be the minimizer of 

\begin{equation*}
\frac{1}{n}\sum_{i=1}^n  \left\{ \widetilde{Y}_i - \sum_{j \leq K+p} \beta_j B_{K, j}(x_i) \right\}^2 + \frac{\lambda}{\mathbb{E} \{\psi^{\prime}(\epsilon)\}} \int_{[a,b]} \left[ \left\{ \sum_{j \leq K+p} \beta_j B_{K, j}(t) \right\}^{(q)} \right]^2 dt .
\end{equation*}
Motivation for the use of pseudo-observations may be found in the linearization of univariate M-estimators that is achieved with the help of the influence function \citep{Maronna:2006}. Note that by A.4 the $\psi(\epsilon_i)$ have mean zero and finite squared expectation. Thus, Theorem 1 of \citet{Claeskens:2009} applies for this theoretical least-squares estimator. With the notation of Section 3, let us finally define the estimating equations

\begin{align}
\boldsymbol{\Phi}(\boldsymbol{\beta}) & = -\frac{1}{n} \sum_{i=1}^n \psi \left(Y_i - \mathbf{B}^{\top}(x_i) \boldsymbol{\beta}  \right) \mathbf{B}(x_i) + 2\lambda \mathbf{D} \boldsymbol{\beta} \\ 
\boldsymbol{\Psi}(\boldsymbol{\beta}) &= -\frac{1}{n}\sum_{i=1}^n \{\widetilde{Y}_i - \mathbf{B}^{\top} (x_i)\boldsymbol{\beta} \}\mathbf{B}(x_i)+ \frac{2\lambda}{ \mathbb{E} \{\psi^{\prime}(\epsilon)\}}  \mathbf{D} \boldsymbol{\beta}  .
\end{align}
The solution to \eqref{eq:7} is a zero of $\boldsymbol{\Phi}$. The zero of $\boldsymbol{\Psi}$, $\widetilde{f}(\cdot) := \sum_j \widetilde{\beta}_j B_{K,j}(\cdot)$, does not correspond to a real estimator, but its implied theoretical properties help in establishing the rates of convergence of $\widehat{f}(\cdot)$ with respect to the semi-norm $||g||_n^2 := n^{-1} \sum_{i=1} |g	(x_i)|^2$. 

For notational simplicity, dependence on $n$ is, in general, suppressed whenever possible. Further, for reasons of convenience both here and in our proofs we also identify each spline function $s(\cdot) = \sum_j \beta_j B_{K,j}(\cdot)$ with its coefficient vector $\boldsymbol{\beta}$. This comes without confusion as all finite-dimensional spaces are isomorphic to Euclidean spaces of equal dimension.

\begin{theorem}
\label{Thm:1}
Let $\widehat{f}(\cdot) = \mathbf{B}^{\top}(\cdot) \widehat{\boldsymbol{\beta}}$ denote a solution of $\boldsymbol{\Phi}(\boldsymbol{\beta}) = \mathbf{0}$. Assume A.1-A.4 and write $C_n := \mathbb{E}\{|| \widetilde{f} - f||_n^2\}$. Then for any $\delta>0$ there exists $n_0$ such that for all $n \geq n_0$

\begin{equation*}
\Pr\left[\text{there is a solution} \ \widehat{f}(\cdot) \ \text{to} \  \boldsymbol{\Phi}(\boldsymbol{\beta}) = \mathbf{0} \ \text{satisfying} \ ||\widehat{f} - \widetilde{f}||_n^2 \leq \delta C_n   \right] \geq 1-\delta.
\end{equation*}
Equivalently, there exists a sequence of M-type penalized spline estimates $\widehat{f}_n(\cdot)$ such that

\begin{equation*}
||\widehat{f} - \widetilde{f}||_n^2/C_n \xrightarrow{P} 0.
\end{equation*}

\end{theorem}

The theorem states that with high probability there exists an M-type penalized spline estimate $\widehat{f}(\cdot)$ such that $\widehat{f}(\cdot)$ and $ \widetilde{f}(\cdot)$ will be much closer than $\widetilde{f}(\cdot)$ and $f(\cdot)$. Theorem \ref{Thm:1} further establishes a useful representation of M-estimators: in a certain sense, M-estimators are equivalent to least-squares estimators applied on the pseudo-observations given in \eqref{eq:16}. This illustrates how different $\rho$-functions operate on the error term: large errors will be either trimmed or discarded based on whether $\rho$ is convex or non-convex with finite rejection point, \citep{Hampel:2011b}.

Additionally, since

\begin{align*}
||\widehat{f} - f ||_n^2  \leq 2 ||\widehat{f} - \widetilde{f} ||_n^2 + 2||\widetilde{f} - f||_n^2 = o_P(C_n) + O_P(C_n)  = O_P(C_n),
\end{align*}
the conclusion is that $\widehat{f}(\cdot)$ will enjoy the same rates of convergence as the least squares estimator $\widetilde{f}(\cdot)$. In particular, let 

\begin{equation*}
K_{q, n} = \widetilde{c}_1^{1/2q}(K_n+p-q)\lambda_n^{1/2q} ,
\end{equation*}
where the constant $\widetilde{c}_1$ is defined in Lemma A3 of \citet{Claeskens:2009} and depends only on $q$ and the design density $w(\cdot)$. The sequence $K_{q, n}$ determines the order of the asymptotic mean squared error, as Theorem \ref{Thm:2} shows.

\begin{theorem}
\label{Thm:2}
Under assumptions A.1-A.4 the following statements hold
\item [(a)] If $K_{q,n}<1$ eventually and $f(\cdot) \in \mathcal{C}^{p}[a,b]$ for $p> 1$, then there exists a sequence of penalized spline M-estimates $\widehat{f}_n(\cdot)$ satisfying
\begin{equation*}
||\widehat{f}_n - f||_n^2 = O_P\left(\frac{K}{n}\right) + O_P \left( \lambda^2 K^{2q} \right) + O_P \left( K^{-2p} \right)
\end{equation*}
\item [(b)] If $K_{q,n} \geq 1$ eventually, $f(\cdot) \in \mathcal{C}^{q}[a,b]$ and $\lim_n n \lambda^{(2q+1)/2q}K^{2q-1} = \lim_n \lambda K^{4q-1} = \infty$, then there exists a sequence of penalized spline M-estimates $\widehat{f}_n(\cdot)$ satisfying 
\begin{equation*}
||\widehat{f}_n- f||_n^2 = O_P\left(\frac{1}{n \lambda^{1/2q}}\right) + O_P \left( \lambda \right) + O_P \left( K^{-2q} \right)
\end{equation*}
\end{theorem}

Theorem \ref{Thm:2} establishes the least-squares mean-squared errors for penalized M-estimators without requiring any moments of the error term. Thus, while infinite error variance would make least-squares estimators inconsistent, M-estimators with bounded $\psi$ functions still maintain their consistency. It should be noted that the conditions $p> 1$ as well as $\lim_n n \lambda^{(2q+1)/2q}K^{2q-1} = \lim_n \lambda K^{4q-1} = \infty$ (in parts (a) and (b), respectively) are purely technical and  are not needed by \citet{Claeskens:2009}. They are important for the M-type estimators because they allow us to control the leverages, which are important quantities in the asymptotics of robust regression estimators with a diverging number of parameters, see \citet{Huber:1973, Yohai:1979, Cox:1983} and the proof of Theorem \ref{Thm:1}.

The rates are extremely interesting because they illustrate the compromise between regression and smoothing splines depending on how fast the number of knots grows, which is essentially what the condition on $K_{q,n}$ entails. In the first case, we have an asymptotic scenario that is similar to that of robust regression splines, cfr. \eqref{eq:14}. The additional term $\lambda^2 K^{2q}$ reflects the shrinkage bias from the penalty parameter and is negligible for a small number of knots. On the other hand, with a larger rate of growth for the number of knots one is led to an asymptotic scenario that is very similar to that of least-squares or M-type smoothing splines, cfr. \eqref{eq:15}. The additional term $K^{-2q}$ is the result of the error of approximation for a $\mathcal{C}^{q}[a,b]$ function by a spline, see the Jackson-type inequality in \citep[Chapter XII]{DB:2001}.  As discussed previously, smoothing splines incur no approximation error as they are exact solutions to the minimization problem \eqref{eq:6}.

Balancing all the MSE components in case (a), by setting $K_n \asymp n^{1/(2p+1)}$ and $\lambda \asymp n^{-\gamma}$ where $\gamma> (p+q)/(2p+1)$, yields $||\widehat{f} - f||_n^2 = O_P(n^{-2p/(2p+1)})$, which is the optimal rate of convergence for $f(\cdot) \in \mathcal{C}^{p}[a,b]$. Similarly, taking $K_n \asymp n^v$, with $v \geq 1/(2q+1)$, and $\lambda_n \asymp n^{-2q/(2q+1)}$ in case (b) yields $|| \widehat{f} - f||_n^2 = O_P(n^{-2q/(2q+1)})$, the optimal rate of convergence for $f(\cdot) \in \mathcal{C}^{q}[a,b]$ \citep{Stone:1982}. Since we have assumed that $ p >q $, the rates of convergence in case (a) are faster than in case (b). As \citet{Claeskens:2009} remark, this fact provides justification for selecting a small number of knots relative to $n$ for penalized spline estimators. Theorem \ref{Thm:2} only asserts the existence of a "good" sequence of estimates. Naturally, this may be strengthened to both existence and uniqueness if one uses a strictly convex $\rho$-function.

The empirical norm $||\cdot||_n$ depends on $x_i$ so that the question arises of whether it is possible to obtain rates of convergence with respect to a global measure. Corollary 1 shows that this is indeed possible with respect to the standard $\mathcal{L}^2([a,b])$ metric denoted by $||\cdot||$. Rates of convergence for the derivatives in the $\mathcal{L}^2([a,b])$ metric may also be obtained.

\begin{cor}
\label{Cor:1}
Assume A.1-A.4 and let $\widehat{f}_{n}(\cdot)$ denote the sequence of penalized M-estimates of Theorem \ref{Thm:1}. Then, if $f \in \mathcal{C}^{p}[a,b], K \asymp n^{1/(2p+1)}$ and $\lambda \asymp n^{-\gamma}$ for $\gamma > (p+q)/(2p+1)$,
\begin{equation*}
||\widehat{f}^{(j)}-f^{(j)}||^2 =  O_P(n^{-(2p-j)/(2p+1)}),
\end{equation*}
while if $f \in \mathcal{C}^{q}[0,1],\ K \asymp  n^{v}$ with $v \geq 1/(2q+1)$ and $\lambda \asymp  n^{-2q/{(2q+1)}}$,
\begin{equation*}
||\widehat{f}^{(j)}-f^{(j)}||^2 =  O_P(n^{-(2q-j)/(2q+1)}).
\end{equation*}
\end{cor}
Corollary \ref{Cor:1} shows that, as one may expect, higher-order derivatives are more difficult to estimate. Unfortunately, this cannot be improved as these rates of convergence are optimal under the weak assumptions of \citet{Stone:1982}.

We now turn to the problem of penalized spline M-estimation with auxiliary scale estimation, such as the IQR of the pseudo-residuals discussed in Section 3. The estimating equations are now modified to 

\begin{align}
\boldsymbol{\Phi}_{\widehat{\sigma}}(\boldsymbol{\beta}) & = -\frac{1}{n} \sum_{i=1}^n \psi \left( \frac{Y_i - \mathbf{B}^{\top}(x_i) \boldsymbol{\beta}}{\widehat{\sigma}}  \right) \frac{\mathbf{B}(x_i)}{\widehat{\sigma}} + 2\lambda \mathbf{D} \boldsymbol{\beta} 
\\ \label{eq:20} \boldsymbol{\Psi}_{\sigma}(\boldsymbol{\beta}) &= -\frac{1}{n}\sum_{i=1}^n \left(\widetilde{Y}_i - \mathbf{B}^{\top}(x_i) \boldsymbol{\beta} \right)\mathbf{B}(x_i)+ 2 \lambda  \frac{\sigma^2}{\mathbb{E} \{\psi^{\prime}(\epsilon/\sigma)\} }  \mathbf{D} \boldsymbol{\beta},
\end{align}
with $\widetilde{Y}_i = f(x_i) + \sigma \psi(\epsilon_i/\sigma)/\mathbb{E}\{\psi^{\prime}(\epsilon/\sigma) \}$.
We aim to show analogues of Theorems \ref{Thm:1} and \ref{Thm:2} for this case. To that end we need a root-$n$ condition on $\widehat{\sigma}$ and a modification of A.4, as follows:

\begin{itemize}
\item[B.4] $\sqrt{n}(\widehat{\sigma}-\sigma) = O_P(1)$ for some scaling constant $\sigma$.
\item[B.5] $\psi(\cdot) \in \mathcal{C}^2(-\infty, \infty)$ with $\sup_{x}|\psi^{\prime \prime}(x)|<\infty$. For any $\epsilon>0$ there exists $M_{\epsilon}<\infty$ such that

\begin{equation*}
| \psi(tx) - \psi(sx)| \leq M_{\epsilon} |t-s|,
\end{equation*} 
for all $t, s > \epsilon$ and $-\infty < x < \infty$. Further, $\mathbb{E} \{\psi(\epsilon/\sigma)\} = 0$, $\mathbb{E} \{\psi^{\prime}(\epsilon/\sigma)\} > 0$, $\mathbb{E}\{|\psi(\epsilon/\sigma)|^2 \} < \infty$ and $\mathbb{E}\{|\psi^{\prime}(\epsilon/ \sigma)|^2 \} < \infty$.
\end{itemize}

The scaling constant $\sigma$ does not need to be the standard deviation of the error term, but it can be when the error has finite variance. Assumption B.5 implies the smoothness of $\psi$ but also that $\psi$ changes slowly in the tail. The latter also implies that $\psi$ is bounded. This condition was also used by \citet{He:1995} and is satisfied, e.g., by redescending $\psi$ functions and Huber $\psi$-functions modulo the smoothness assumption, see the remarks following A.4. The moment conditions parallel those of A.4 and ensure the Fisher-consistency of the estimates.

With these assumptions we can prove that there still exists a sequence of M-estimates that achieves the optimal rates of convergence.

\begin{theorem}
\label{Thm:3}
Under assumptions A.1-A.3 and B.4-B.5 there exists a sequence of M-type penalized spline estimates $\widehat{f}_n(\cdot)$ solving $\boldsymbol{\Phi}_{\widehat{\sigma}}(\boldsymbol{\beta}) = \mathbf{0}$ for which the following statements hold:
\item [(a)] If $f(\cdot) \in \mathcal{C}^{p}[a,b], K \asymp n^{1/(2p+1)}$ and $\lambda \asymp n^{-\gamma}$ with $ \gamma > (q+p)/(2p+1)$,
\begin{equation*}
||\widehat{f}-f||_n^2 = O_P( n^{-2p/(2p+1)} ).
\end{equation*}
\item[(b)] If $f(\cdot) \in \mathcal{C}^{q}[a,b],\ K \asymp  n^{v}$ with  $v  >1/(2q+1)$ and $\lambda \asymp  n^{-2q/{(2q+1)}}$,
\begin{equation*}
||\widehat{f}-f||_n^2 = O_P( n^{-2q/(2q+1)}).
\end{equation*}
\end{theorem}
\noindent 
Note that the rate of growth of the knots and the rates of decay for $\lambda$ in each one of these cases ensure that we eventually have either $K_{q, n} <1 $ or $K_{q, n} \geq 1$. Corollary 1 regarding rates of convergence of derivatives in the $\mathcal{L}^2([a,b])$ metric carries over to this case in the same manner as previously. In general, scale estimators constructed using linear combinations of $Y_{i}$, such as those discussed in Section 3.1, satisfy the root-n condition and thus provide good preliminary scale estimates.

\section{A Monte-Carlo study}

In our simulation experiments we compare the performance of the M-type penalized spline estimator with auxiliary scale to the least-squares penalized spline and smoothing spline estimators.  For the penalized M-estimator we use the Huber $\rho$-function with tuning parameter equal to $1.345$, corresponding to 95\% efficiency in the location model. The auxiliary scale estimate is the IQR of the pseudo-residuals and we select the smoothing parameters for the penalized spline M-estimator as outlined in Section 3.2. 
The least-squares estimators can be easily fitted with the \texttt{gam} function \citep{Wood:2019} in the freeware \textsf{R}, \citep{R:2019}. By default, the penalty parameter is estimated by restricted maximum likelihood, see \citet{Ruppert:2003} and \citet{Wood:2017} for more details on this technique. 

We investigate the performance of the estimators in the regression model $Y_i = f(t_i) + \epsilon_i$ where $t_i = i/n$ and $f(\cdot)$ is each of the following three functions

\begin{enumerate}
\item $f_1(t) = \sin(2 \pi t) + \exp\{ -3(t-0.5)^2 \} + 0.4$,
\item $f_2(t) = 1/(0.1+t)+8\exp\{-400(t-0.5)^2\}$,
\item $f_3(t) = \phi((t-0.5)/0.15)- \phi((t-0.8)/0.04)$,
\end{enumerate}
where $\phi(\cdot)$ denotes the Gaussian density. All three functions are smooth but they differ qualitatively as $f_2(\cdot)$ and $f_3(\cdot)$, in contrast to $f_1(\cdot)$, exhibit strong local characteristics in the form of spikes and bumps.

In order to assess the effect of outliers on the estimates different distributions for the error term were considered. Other than the standard Gaussian distribution, we have complemented our set-up with a t-distribution with 3 degrees of freedom, a mixture of mean-zero Gaussians with standard deviations equal to 1 and 9 and weights equal to 0.85 and 0.15 respectively, as well as Tukey's Slash distribution. The resulting mean-squared-errors are summarized in Table \ref{tab:1} for sample sizes of 100 and 1000 replications.

\begin{table}[H]
\centering
\begin{tabular}{l c c c c c c c c }
\hline
& & \multicolumn{2}{c}{Huber(Psp)} & \multicolumn{2}{c}{Least-squares(Psp)} & \multicolumn{2}{c}{Least-squares(Smsp)} \\ 
$f(\cdot)$ & Distribution & Mean & Median & Mean & Median & Mean & Median  \\ \cline{1-8}

\multirow{4}{*}{$f_1(\cdot)$} & Gaussian & 0.067  & 0.055 & 0.068 & 0.052 & 0.062 & 0.049 \\
 & $t_3$ & 0.100 &  0.079 & 0.186 & 0.129 & 0.190 & 0.184 \\
 & M. Gaussian &  0.144 & 0.107 & 0.714 & 0.454 & 0.701 & 0.418 \\ 
 & Slash & 0.968 & 0.355 & 4979.6 & 5.761 & 5084.5 & 6.012 \\ \cline{1-8}
\multirow{4}{*}{$f_2(\cdot)$} & Gaussian &  0.225 & 0.217 & 0.203 & 0.198 & 0.212 & 0.198 \\
 & $t_3$ & 0.359  & 0.323  & 0.487&  0.419 & 0.515 & 0.430 \\
 & M. Gaussian & 0.699  & 0.535  & 1.772 & 1.594 & 1.883 & 1.583	 \\ 
 & Slash & 4.051 & 3.029 & 2476.7  & 8.506 &  2790.2 &  9.341 \\ \cline{1-8}
\multirow{4}{*}{$f_3(\cdot)$} & Gaussian & 0.056  & 0.045  & 0.064 & 0.048 & 0.062 &  0.047  \\
 & $t_3$ &  0.080&  0.056& 0.154 & 0.092  & 0.143 & 0.087  \\
 & M. Gaussian & 0.079 &0.061 & 0.665  & 0.350 & 0.550 & 0.306 \\ 
 & Slash & 0.481  & 0.146 &  1892.39 &  5.998 &  2604.5 &  5.489 \\ \cline{1-8}
\end{tabular}
\caption{Means and medians of the 1000 MSEs of the penalized spline Huber M-estimator, the penalized spline least-squares estimator and the smoothing spline least-squares estimator.}
\label{tab:1}
\end{table}

Comparing the two least-squares estimators reveals that using fewer basis functions than the number of observations hardly impacts performance. The results further demonstrate the extreme sensitivity of least-squares estimators with respect to even a small number of aberrant observations. We note, in particular, that both the t-distribution with 3 degrees of freedom and the mixture of Gaussian distributions have first moments equal to zero and finite second moments yet the performance of least-squares estimators is unduly affected. By contrast, the Huber estimator matches the performance of the least-squares estimator in Gaussian data and exhibits a large degree of resistance to aberrant observations. These facts illustrate the favourable trade-off penalized M-estimators tend to achieve.

\section{Real data examples}

\subsection{Mid-atlantic wage data}

We now illustrate the proposed M-type penalized spline estimators on two real datasets: the mid-atlantic wage dataset and the mammals dataset.  For the purpose of comparison we also include the least-squares estimator.  The datasets are freely available in the \textsf{R}-packages \texttt{ISLR} \citep{James:2013} and \texttt{quantreg} \citep{Koenker:2012} respectively.

The mid-atlantic wage dataset consists of 3000 observations on different characteristics of male workers in the said region of the United States. The dataset contains eleven socio-economic variables but here we focus on the relationship between age in years and yearly raw wages recorded in 2011 US dollars. Typically, income distribution is right-skewed so a few outlying observations would be the rule rather than the exception. A scatter-plot of these variables with the Huber and the least-squares penalized fits is shown in the left panel of Figure \ref{fig:1}.

\begin{figure}[H]
\centering
\subfloat{\includegraphics[width = 0.495\textwidth]{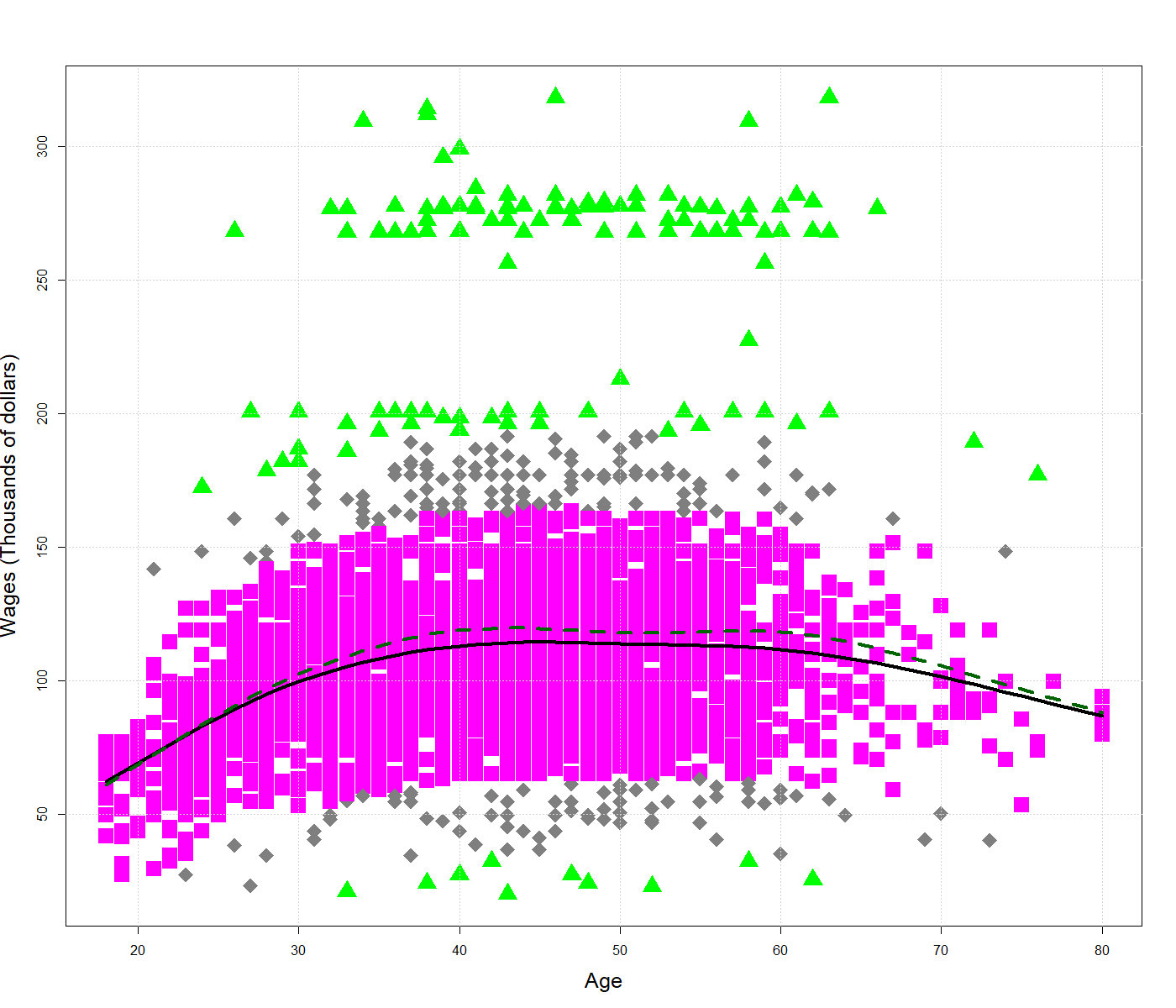}} \
\subfloat{\includegraphics[width = 0.495\textwidth]{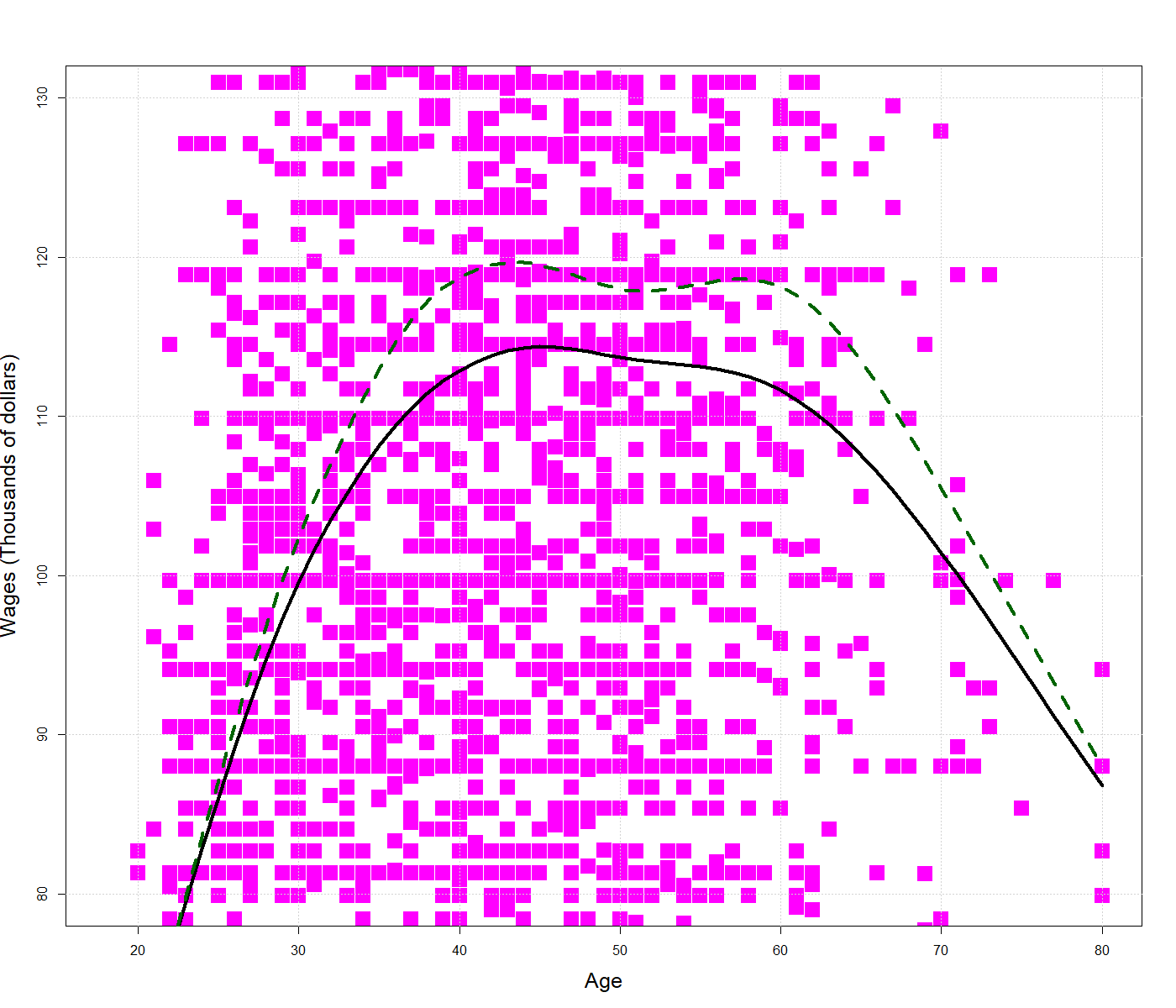}}
\caption{Left: scatter plot of the mid-atlantic wage dataset with least-squares (dashed curve) and Huber (solid curve) fits superimposed. Right: A closer look at the fits with a rescaled vertical axis. The symbols $\blacktriangle, \blacklozenge, \blacksquare$ correspond to observations with weights $(0, 0.33],(0.33, 0.66], (0.66, 1]$ respectively.}
\label{fig:1}
\end{figure}

Both fits point towards a curvilinear relationship with a slight bend downwards; this reflects the fact that workers reach their peak income at the middle of their work life and their income slightly decreases afterwards. Some care, however, is needed in this interpretation as relatively few observations are available for workers past the usual retirement age.

Comparing the fits in more detail leads us to notice that the least-squares fit overestimates the mean salary for younger to middle-aged workers; this may be explained by observing that a number of very high-earners exert disproportionate influence on the estimate, effectively pulling the fit upwards. By contrast, the M-estimator remains largely resistant to these atypical observations. To better understand this discrepancy the plot also includes color and shape coding based on the weights generated by the M-estimator. While all observations receive equal weights by the least-squares estimator, the observations corresponding to atypically high-earners are greatly down-weighted by the M-estimator leading to differences in the fits. Restricting attention to observations in the middle on the right panel shows that these differences can be as high as 10000 US dollars, which is a respectable amount of money from both individual and policy standpoints.

\subsection{Mammals weight and speed data}

The mammals dataset consists of 107 observations on maximal running speeds and weights of mammals, see \citet{Garland:1983} for more information as well as a parametric regression analysis of this relationship. A scatter plot of these variables with the Huber and the least-squares penalized fits is shown in the left panel of Figure \ref{fig:2}.

\begin{figure}[h]
\centering
\subfloat{\includegraphics[width = 0.495\textwidth]{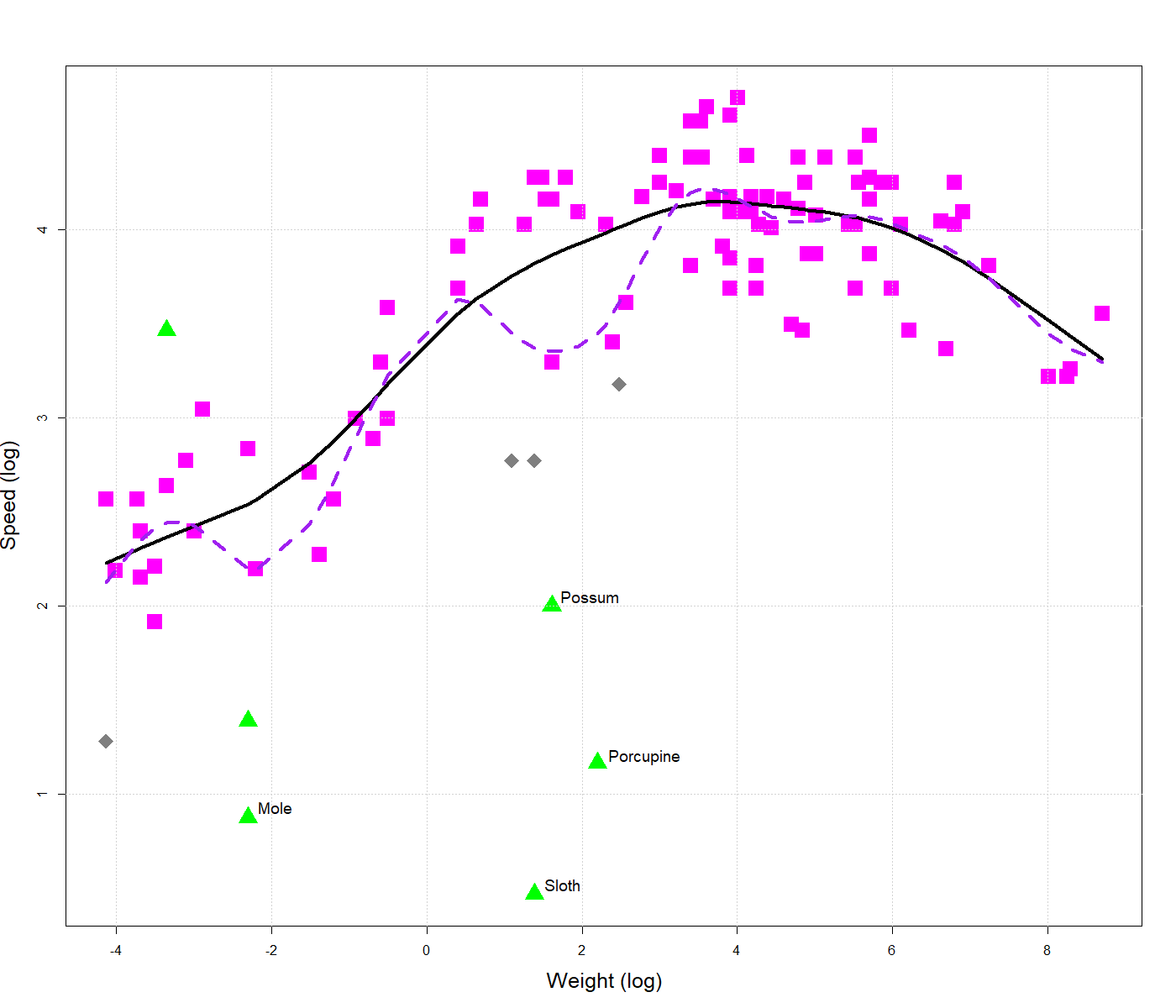}} \
\subfloat{\includegraphics[width = 0.495\textwidth]{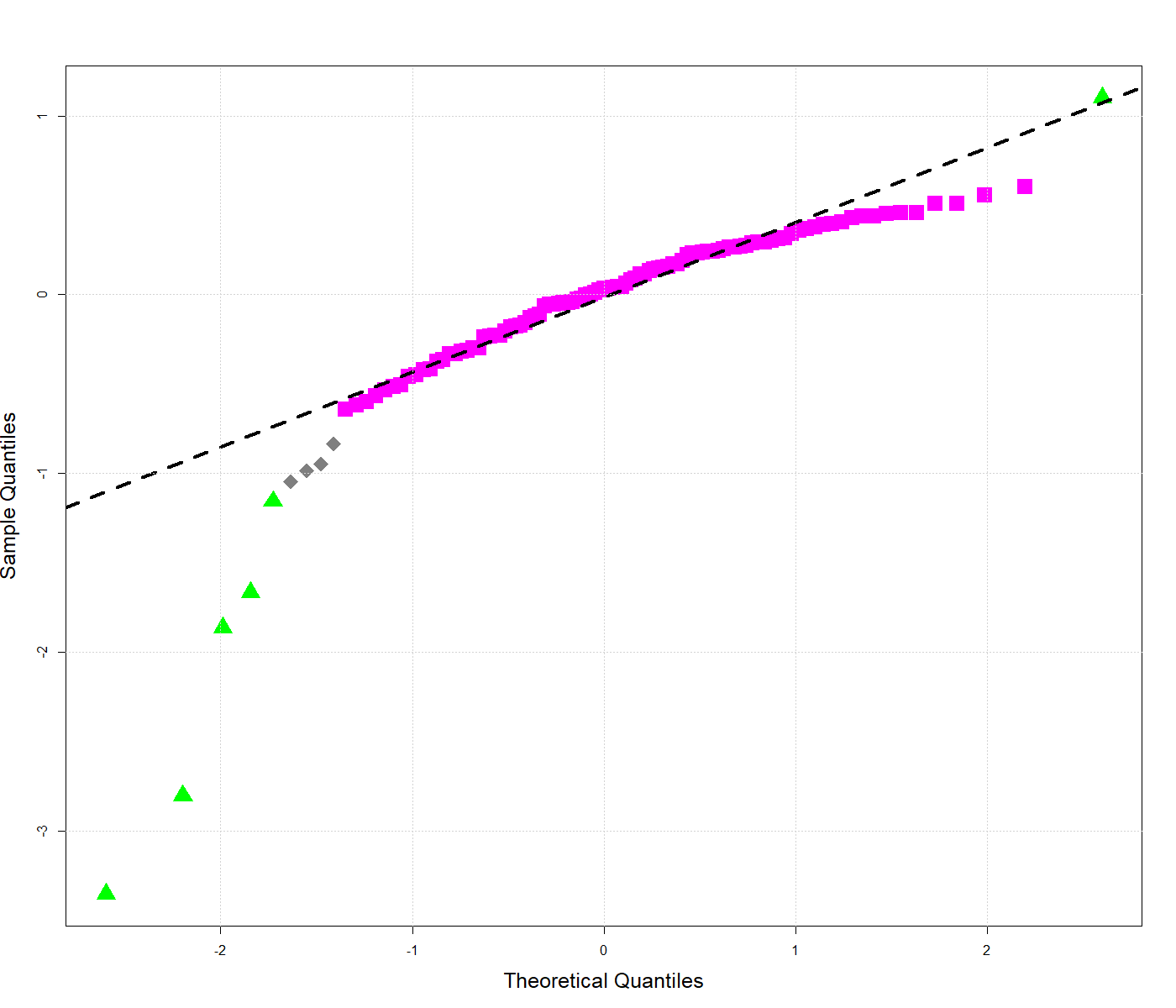}}
\caption{Left: scatter plot of the mammals dataset with least-squares (dashed curve) and Huber (solid curve) fits superimposed. Right: QQ plot of the residuals of the Huber fit. The symbols $\blacktriangle, \blacklozenge, \blacksquare$ correspond to observations with weights $(0, 0.33],(0.33, 0.66], (0.66, 1]$ respectively.}
\label{fig:2}
\end{figure}

The scatter plot and the fits easily refute the naive hypothesis that speed should be a decreasing function of weight. Curiously, speed seems to increase with weight up to a certain point, which we may call "optimal weight", and decrease afterwards. That is, neither the smallest nor the largest animals are the fastest.

Several outliers significantly impact the least-squares analysis, however. These for the most part correspond to small rodent-like animals (see the included labels) whose speed is far below what could be expected given their weight. The effect of these outlying observations is again to pull the least-squares fit towards them resulting in substantial added curvature.  On the contrary, the 4 most outlying observations receive a near-zero weight by the M-estimator and as a result their impact on the estimated regression function is limited. The QQ plot in the right panel indicates that rather than trying to fit all observations, the M-estimator only focuses on the "good" majority for which it provides a better fit than the least-squares estimator. 

Overall, both examples illustrate the benefits of using an M-type penalized spline estimator for the analysis of data with atypical observations. 

\section{Discussion}

The results in this paper indicate that there is little theoretical difference between least-squares and general M-type penalized spline estimators. In particular, the findings of \citet{Claeskens:2009} in support for a smaller number of knots also apply to M-type penalized spline estimators. The latter class of estimators is broad enough to include the least-squares estimator as a special case but also includes estimators that are much less susceptible to atypical observations while performing as well as the least-squares estimators in clean data sets. For these M-type estimators, under some weak restrictions the results carry over even if one uses a preliminary scale estimate as a means of standardization. This can be useful for outlier detection, as demonstrated in our two real-data examples. In effect, we view the proposed penalized spline estimator as the first of its kind combining good theoretical properties and computational ease.

It would be of great interest to extend the penalized spline estimation techniques presented here to robust generalized linear models, using, e.g., the estimator proposed by \citet{ Cantoni:2001}. Generalized additive models have been immensely popular in recent years due to their flexibility and ease of use and we are confident that M-type penalized spline estimation can be successfully used in this context as well. Another important area where robust penalized sieved estimators can be successful is functional regression and its variants, which have also  attracted great interest recently. We aim to study such extensions in detail as a part of our future work.

\section*{Software availability}

An implementation of the M-type penalized spline estimator with settings described herein may be found in the website \url{https://wis.kuleuven.be/statdatascience/robust/papers-2010-2019}. The code also reproduces the plots of Section 6.

\section*{Appendix}

The Appendix contains the proofs of Theorem \ref{Thm:1}, Corollary \ref{Cor:1} and Theorem \ref{Thm:3}. A proof that the computational algorithm converges to a stationary point of the objective function is available from the authors upon request. 

We start with two lemmas that are crucial for the proofs of the results of Section 4. Lemma \ref{Lem:1} essentially states that splines are excellent approximators of smooth functions and Lemma \ref{Lem:2} establishes a set of strong Lindeberg-type conditions on the rows of the spline-design matrix.

\begin{lemma}
\label{Lem:1}
For each $f(\cdot) \in \mathcal{C}^{j}[a,b]$ there exists a spline function $s_f$ of order $p$, $p >j$ such that

\begin{equation*}
\sup_{x \in [a,b]} | f(x) - s_f(x) | \leq \const_{p,j} |\mathbf{t}|^{j} \sup_{|x-y|<\mathbf{t}} |f^{(j)}(x)-f^{(j)}(y)|
\end{equation*}
where $\mathbf{t} = \max_i|t_i - t_{i-1}|$ and the constant depends only on $p$ and $j$.
\end{lemma}
\begin{proof}
See \citet[pp.145-149]{DB:2001}. 
\end{proof}

\begin{lemma}
\label{Lem:2}
Define $\mathbf{G} := n^{-1} \sum_{i=1}^n \mathbf{B}(x_i) \mathbf{B}^{\top}(x_i) + \lambda \mathbf{D}$ with $\mathbf{B}(x_i)$ and $\mathbf{D}$ defined in Section 3.1. 
Under A.1-A.3 the following asymptotic relations hold
\begin{itemize}
\item[(i)] $n^{-1} \max_{1 \leq i \leq n} \mathbf{B}^{\top} (x_i)\mathbf{G}^{-1} \mathbf{B}(x_i) = o(1)$, as $n \to \infty$.
\item[(ii)] $C_n \max_{1 \leq i \leq n} \mathbf{B}^{\top}(x_i) \mathbf{G}^{-1} \mathbf{B}(x_i)= o(1)$ as $n \to \infty$,
\item[(iii)] $\max_{1 \leq i \leq n} \mathbf{B}^{\top}(x_i) \mathbf{G}^{-1} \mathbf{B}(x_i)/(C_n n) = O(1)$ as $n \to \infty$, under the conditions of Theorem \ref{Thm:3}.
\end{itemize}
where $C_n = \mathbb{E} \{|| \widetilde{f}- f||_n^2\}$ is the average mean-squared error of the theoretical least-squares estimator.

\end{lemma}

\begin{proof}
It suffices to prove $(ii)$ and $(iii)$ as $(i)$ follows directly from $(ii)$ given that $C_n$ converges to zero at a rate strictly lower than the parametric rate. 

Assume first that $K_{q,n}<1$ eventually. Either by Lemma 6.2 of \citet{Shen:1998} or by Lemma 5.1 of \citet{Shi:1995} it follows that for large $n$ there exists a positive constant $c$ such that 

\begin{equation*}
\lambda_{\min}(n^{-1} \sum_{i=1}^n \mathbf{B}(x_i) \mathbf{B}^{\top}(x_i)) \geq c K^{-1},
\end{equation*}
where $\lambda_{\min}(\cdot)$ denotes the function that returns the smallest eigenvalue of a symmetric positive semi-definite matrix. Further, since for any $\boldsymbol{\beta} \in \mathbb{R}^{K+p}$,	

\begin{equation*}
0 \leq \int_{[a,b]} \left[ \left\{ \sum_{j \leq K+p} \beta_j B_{K, j}(t) \right\}^{(q)} \right]^2 dt = \boldsymbol{\beta}^{\top} \mathbf{D} \boldsymbol{\beta},
\end{equation*}
the penalty matrix $\mathbf{D}$ is positive semi-definite and, as $\lambda>0$, we find that
\begin{align*}
\max_{1 \leq i \leq n} \mathbf{B}^{\top}(x_i) \mathbf{G}^{-1} \mathbf{B}(x_i) & \leq \max_{1 \leq i \leq n} \mathbf{B}^{\top} (x_i) \left\{n^{-1} \sum_{j=1}^n \mathbf{B}(x_j) \mathbf{B}^{\top} (x_j) \right\}^{-1} \mathbf{B}(x_i) 
\\ & \leq c^{-1} K \max_{1 \leq i \leq n} \sum_{\ell=1}^{K+p} |B_{k,\ell}(x_i)|^2 \\
& \leq c^{-1} K \max_{1 \leq i \leq n} \underbrace{\max_{1 \leq \ell \leq K+p} B_{k, \ell}(x_i)}_{\leq 1} \underbrace{\sum_{\ell=1}^{K+p} B_{k,\ell}(x_i)}_{=1}
\\ & \leq c^{-1} K,
\end{align*}
by properties (b) and (c) of the B-spline functions. Result (ii) now follows from the expression for $C_n$ for $K_{q,n}<1$ of Theorem 2 and the additional assumption that $p >1$. 

In the case that $K_{q,n} \geq 1 $ eventually  we need a tighter bound for the smallest eigenvalue $\mathbf{G}$, as given in Lemma A1 of \citet{Claeskens:2009}. In particular, for $K_{q,n} \geq 1$ these authors show that with probability tending to one

\begin{equation*}
\lambda_{\min}(\mathbf{G}) \geq c K^{-1} (1+ K_{q,n}^{2q} ),
\end{equation*}
for the same constant $c$. Since $K_{q,n}^{2q} \sim \widetilde{c_1} K^{2q} \lambda$, to check (ii) it suffices to establish that

\begin{equation*}
\left( \frac{1}{n\lambda^{1/2q}} +  \lambda +  K^{-2q} \right) K^{1-2q} \lambda^{-1} = o_P(1),
\end{equation*}
as $n \to \infty$. This holds by assumption. In particular, it holds for $\lambda \asymp  n^{-2q/(2q+1)}$ and $K \asymp n^{v}$ with $v \geq 1/(2q+1)$ that lead to the optimal rates of convergence.	

To prove (iii) note that either $C_n \asymp n^{-2p/(2p+1)}$ or $C_n \asymp n^{-2p/(2p+1)}$ under the conditions of Theorem \ref{Thm:3}, that is, the rates of growth of $K$ and the rates of decay of $\lambda$. The statement follows from simple multiplication.
	
\end{proof}

We now turn to the proof of Theorem 1. 

\begin{proof}[Proof of Theorem 1]

From Lemma 1 there exists a spline function $f^{\star}$ such that

\begin{equation*}
\sup_{x \in [a,b]} |f(x) - f^{\star}(x)| = O(D_K),
\end{equation*}
where $D_K = K^{-p}$ or $D_K = K^{-q}$ depending on whether $f \in \mathcal{C}^{p}[a,b]$ or $f \in \mathcal{C}^{q}[a,b]$ respectively. Since B-splines form a basis for $S_{K}^p$ we can put $f^{\star} = \sum_{j=1}^{K+p} \beta^{\star}_j B_{K,j}$. 

Expanding $\mathbf{\Psi}$ we can write

\begin{equation*}
\boldsymbol{\Psi}(\boldsymbol{\beta}) = -\frac{1}{n} \sum_{i=1}^{n} \left\{  f(x_i) -  \mathbf{B}^{\top}(x_i) \boldsymbol{\beta} \right\} \mathbf{B}(x_i) - \frac{1}{n \mathbb{E} \{\psi^{\prime}(\epsilon)\} } \sum_{i=1}^n \psi(\epsilon_i) \mathbf{B}(x_i) + 2 \frac{\lambda}{\mathbb{E} \{\psi^{\prime}(\epsilon)\}} \mathbf{D} \boldsymbol{\beta}.
\end{equation*}
Let $\lambda_0 := 2\lambda/\mathbb{E} \{\psi^{\prime}(\epsilon)\}$ and define, as before,
\begin{equation*}
\mathbf{G}:= n^{-1}\mathbf{B}^{\top} \mathbf{B} + \lambda_0 \mathbf{D}.
\end{equation*}
Recall that $\psi(\cdot) \in \mathcal{C}^{2}(-\infty,\infty)$, by A.4. A Taylor expansion allows us to write

\begin{align*}
\sum_{i=1}^n \psi \left(Y_i - \mathbf{B}^{\top}(x_i) \boldsymbol{\beta}  \right) \mathbf{B}(x_i) & = \sum_{i=1}^n \psi \left( \epsilon_i + f(x_i) - \mathbf{B}^{\top}(x_i) \boldsymbol{\beta}  \right) \mathbf{B}(x_i) \\
& = \sum_{i=1}^{n} \psi(\epsilon_i) \mathbf{B}(x_i) + \sum_{i=1}^n \psi^{\prime} (\epsilon_i) \left\{ f(x_i) - \mathbf{B}^{\top}(x_i) \boldsymbol{\beta}  \right\} \mathbf{B}(x_i) \\ &\phantom{{}=1}  + \frac{1}{2}\sum_{i=1}^n \psi^{\prime \prime} (c_i) \left\{ f(x_i) - \mathbf{B}^{\top}(x_i) \boldsymbol{\beta} \right\}^2 \mathbf{B}(x_i).
\end{align*}
for some mean values $c_i$, each depending on the $i$th observation.	

On $\mathbb{R}^{K+p}$ define the bilinear form
\begin{align*}
\left\| \boldsymbol{\beta} \right\|_{\mathbf{G}}^2 & : = n^{-1} ||\mathbf{B} \mathbf{G}^{-1}\boldsymbol{\beta}||_{E}^2,
\end{align*}
where $|| \cdot ||_{E}$ denotes the usual euclidean norm.  We note that$ || \cdot ||_{\mathbf{G}}$ is well-defined because $\mathbf{G}$ is invertible for large $n$, see the proof of Lemma 1. From the triangle inequality we have

\begin{align*}
\left\|\boldsymbol{\Psi}(\boldsymbol{\beta}) - \boldsymbol{\Phi}(\boldsymbol{\beta})/\mathbb{E} \{\psi^{\prime}(\epsilon)\} \right\|_{\mathbf{G}} \leq T_1 + T_2,
\end{align*}
with

\begin{align*}
T_1 & = \left\| \frac{1}{n \mathbb{E} \{\psi^{\prime}(\epsilon)\}}  \sum_{i=1}^n \psi^{\prime} (\epsilon_i) \left\{ f(x_i) - \mathbf{B}^{\top}(x_i) \boldsymbol{\beta} \right\} \mathbf{B}(x_i) - \frac{1}{n} \sum_{i=1}^n \left\{ f(x_i) - \mathbf{B}^{\top}(x_i) \boldsymbol{\beta}  \right\} \mathbf{B}(x_i)\right\|_{\mathbf{G}} \\
& = \left\|\frac{1}{n \mathbb{E} \{\psi^{\prime}(\epsilon)\} }  \sum_{i=1}^n  \left\{ f(x_i) - \mathbf{B}^{\top} (x_i) \boldsymbol{\beta} \right\} \left\{ \psi^{\prime}(\epsilon_i) - \mathbb{E} \{\psi^{\prime}(\epsilon)\} \right\} \mathbf{B}(x_i) \right\|_{\mathbf{G}},
\end{align*}
and 

\begin{align*}
T_2 = \left\| \frac{1}{2 n \mathbb{E} \{\psi^{\prime}(\epsilon)\} }\sum_{i=1}^n \psi^{\prime \prime} (c_i) \left\{ f(x_i) - \mathbf{B}^{\top}(x_i) \boldsymbol{\beta} \right\}^2 \mathbf{B}(x_i) \right\|_{\mathbf{G}}.
\end{align*}
Using A.2, A.4 as well as the independence and identical distributions of $\epsilon_i$, we have

\begin{align}
\label{eq:21}
\mathbb{E} T_1^2 & =  \frac{\Var \{ \psi^{\prime}(\epsilon) \}}{n^2 \left(\mathbb{E} \{\psi^{\prime}(\epsilon)\} \right)^2 } \sum_{i=1}^n \left\{ f(x_i) - \mathbf{B}^{\top}(x_i) \boldsymbol{\beta} \right\}^2 \left\|  \mathbf{B}(x_i) \right\|_\mathbf{G}^2 \nonumber
\\ & \leq   n^{-1}\max_{1\leq i\leq n}\left\|  \mathbf{B}(x_i) \right\|_\mathbf{G}^2 \frac{\mathbb{E}\{|\psi^{\prime}(\epsilon)|^2 \}}{\left(\mathbb{E} \{\psi^{\prime}(\epsilon)\} \right)^2 } n^{-1} \sum_{i=1}^n \left\{f(x_i) - \mathbf{B}^{\top}(x_i) \boldsymbol{\beta} \right\}^2.
\end{align}
The second term, $T_2$, can now be estimated similarly to give the bound

\begin{align}
T_2 & \leq \frac{\sup_ |\psi ^{\prime \prime}(x)|}{2 \mathbb{E} \{\psi^{\prime}(\epsilon)\}} n^{-1} \sum_{i=1}^n \left\{ f(x_i) - \mathbf{B}^{\top}(x_i) \boldsymbol{\beta} \right\}^2 \left\| \mathbf{B}(x_i) \right\|_\mathbf{G} \nonumber\\
& \leq  \frac{\sup_x |\psi ^{\prime \prime}(x)|}{2 \mathbb{E} \{\psi^{\prime}(\epsilon)\}}  \max_{1 \leq i\leq n} \left\|\mathbf{B}(x_i) \right\|_\mathbf{G} n^{-1} \sum_{i=1}^n  \left\{f(x_i) - \mathbf{B}^{\top}(x_i) \boldsymbol{\beta}\right\}^2 \nonumber,
\end{align}
where we again have used A.4.

Now let $f^{\star} = \sum_j \beta_j B_{K,j}$ denote the spline approximation to $f$ constructed with the help of Lemma 1 and let $\widetilde{f}(\cdot) = \sum_j \widetilde{\beta}_j B_{K,j}(\cdot)$ denote the zero of $\mathbf{\Psi}$. We have

\begin{equation}
|| \widetilde{f}-f^{\star}||_n^2 \leq 2 || \widetilde{f}-f||_n^2 + 2 || f-f^{\star}||_n^2  = O_P(C_n)+ O(D_K)
\end{equation}
where $C_n = \mathbb{E}\{ ||\widetilde{f}-f ||_n^2\}$. Since the approximation error $D_K$ is included in $C_n$, see Theorem 1 of \citet{Claeskens:2009}, we conclude that $|| \widetilde{f}-f^{\star}||_n^2 = O_P(C_n)$ and hence there exists a constant $K_1(\delta)$ such that $|| \widetilde{f}-f^{\star}||_n \leq 1/2 (K_1 C_n)^{1/2}$ for all large $n$ with probability greater than $1-\delta/4$. 

Define the sets

\begin{equation}
F_n := \left\{ s\in S_{K}^p : ||s-f^{\star}||_n^2 \leq K_1 C_n \right\}.
\end{equation}
Letting $B := (8 \delta^{-1} \Var \{\psi^{\prime}(\epsilon)\} (\mathbb{E} \{\psi^{\prime}(\epsilon) \})^{-2} )^{1/2}$, Markov's inequality and \eqref{eq:21} imply that 

\begin{equation*}
T_1 \leq B \left( n^{-1} A_{n}^2 ||f-s||_n^2 \right)^{1/2},
\end{equation*}
where $A_n^2 = \max_{i\leq n} ||\mathbf{B}(x_i)||_{\mathbf{G}}^2$, with probability greater than $1-\delta/8$. Working now on $F_n$, by the previous decomposition it follows that there exists a constant $K_2(\delta)$ such that $||s-f||_n^2 \leq K_2 C_n$ with probability also greater than $1-\delta/8$. Combining these two events we see that with probability greater than $1-\delta/4$ 

\begin{equation}
T_1 \leq B \left( n^{-1} K_2 C_n A_{n}^2 \right)^{1/2},
\end{equation}
for all large $n$. If $s \in F_n$ and we set $B^{\prime} := (2 \mathbb{E} \{\psi^{\prime}(\epsilon)\})^{-1} \sup_x | \psi^{\prime \prime}(x)|$ we also have
\begin{equation}
T_2 \leq B^{\prime} K_2 A_n C_n
\end{equation}
with probability greater than $1-\delta/8$ for all large $n$. 

Combining all the above bounds we obtain that for large $n$
\begin{align*}
|| \boldsymbol{\Phi} (\boldsymbol{\beta})/\mathbb{E} \{\psi^{\prime}(\epsilon)\} - \mathbf{G} (\boldsymbol{\beta} - \boldsymbol{\beta}^{\star}) ||_{_{\mathbf{G}}} \nonumber & \leq || \boldsymbol{\Phi} (\boldsymbol{\beta})/\mathbb{E} \{\psi^{\prime}(\epsilon)\} -  \boldsymbol{\Psi} (\boldsymbol{\beta}) ||_{\mathbf{G}} \\ &\phantom{{}=1} + ||  \mathbf{G}(\boldsymbol{\beta}-\tilde{\boldsymbol{\beta}})  -\mathbf{G}(\boldsymbol{\beta} - \boldsymbol{\beta}^{\star}) ||_{\mathbf{G}} \nonumber
\\ & \leq B( n^{-1}  K_2 C_n A_n^2)^{1/2} + B^{\prime} K_2 C_n A_n  + 2^{-1} (K_1 C_n)^{1/2} \nonumber
\end{align*}
on an event with probability greater than $1-\delta$, the first inequality following from the fact that 
\begin{equation*}
\mathbf{\Psi}(\boldsymbol{\beta}) = \mathbf{G} \boldsymbol{\beta} - n^{-1}\mathbf{B}^{\top} \mathbf{Y} = \mathbf{G}\boldsymbol{\beta} - \mathbf{G}\widetilde{\boldsymbol{\beta}},
\end{equation*}
as $\widetilde{\boldsymbol{\beta}}$ is the zero of $\mathbf{\Psi}$ and the second inequality from all previous probabilistic	bounds. Factoring $(K_1 C_n)^{1/2}$ we may certainly write
\begin{multline}
\label{eq:26}
 || \boldsymbol{\Phi} (\boldsymbol{\beta})/\mathbb{E} \{\psi^{\prime}(\epsilon)\} - \mathbf{G} (\boldsymbol{\beta} - \boldsymbol{\beta}^{\star}) ||_{_{\mathbf{G}}} \leq \\  \left\{  B( K_2 K_1^{-1} n^{-1} A_n^2)^{1/2} + B^{\prime} K_2 K_1^{-1/2} A_{n} C_n^{1/2} + 2^{-1} \right\} (K_1 C_n)^{1/2}.
\end{multline}
Further, since $\mathbf{D}$ is positive semidefinite and $\lambda_0>0$,

\begin{align}
\label{eq:27}
A_n^2 = \max_{1\leq i\leq n} ||\mathbf{B}(x_i) ||_\mathbf{G}^2 & = \max_{1\leq  i\leq n} n^{-1}\sum_{j=1}^n \left\{\mathbf{B}^{\top}(x_j) \mathbf{G}^{-1} \mathbf{B}(x_i) \right\}^2 \nonumber
\\ & = \max_{1 \leq i \leq n} \mathbf{B}^{\top}(x_i) \mathbf{G}^{-1} \left\{ n^{-1} \sum_{j=1}^n \mathbf{B}(x_j) \mathbf{B}^{\top}(x_j) \right\} \mathbf{G}^{-1} \mathbf{B}(x_i) \nonumber
\\ & \leq \max_{1 \leq i \leq n} \mathbf{B}^{\top}(x_i) \mathbf{G}^{-1} \mathbf{B}(x_i),
\end{align}
Lemma 2 allows us to deduce that $\lim_n A_n^2 C_n = \lim_n n^{-1} A_n^2  = 0$. This in turn means that the term inside the curly brackets in \eqref{eq:26} will be smaller than $1$ for $n$ sufficiently large. 

For such $n$ if $s \in F_n - f^{\star}$ with coefficient vector $\boldsymbol{\beta}$ and we define

\begin{equation}
\boldsymbol{U} (\boldsymbol{\beta}) : = \boldsymbol{\beta}- \mathbf{G}^{-1}\boldsymbol{\Phi}(\boldsymbol{\beta} + \boldsymbol{\beta}^{\star})/\mathbb{E} \{\psi^{\prime}(\epsilon)\}
\end{equation}
then on account of \eqref{eq:26} we must have  $||\mathbf{B} \mathbf{U}(\boldsymbol{\beta})||_{E}^2 \leq K_1 C_n $ for all large $n$. The set $F_n - f^{\star}$ is clearly convex. We claim that for sufficiently large $n$ it is also compact. Indeed, the set is finite-dimensional, closed and because the matrix $n^{-1} \mathbf{B}^{\top} \mathbf{B}$ is nonsingular for large $n$, see the proof of Lemma 2, it is also bounded. Hence the claim holds.

We now see that $\boldsymbol{U} (\boldsymbol{\beta})$ is a continuous function mapping the compact, convex set $F_n-f^{\star}$ into itself. Thus, Brouwer's theorem assures us of the existence of a fixed point $s^{\prime}$ in $F_n - f^{\star}$. Putting $\widehat{f} := s^{\prime} + f^{\star}$, it is easily seen that $ \boldsymbol{\Phi}( \widehat{\boldsymbol{\beta}} ) = \mathbf{0}$, i.e. $\widehat{\boldsymbol{\beta}}$ is the zero of the estimating equation . 

By the above, and since $n^{-1}||\mathbf{B}\boldsymbol{\widehat{\beta}}-\mathbf{B}\boldsymbol{\widetilde{\beta}}||_E^2 = ||\widehat{f} - \widetilde{f}||_n^2$, it now follows

\begin{align*}
||\widehat{f} - \widetilde{f} ||_n & =
|| \boldsymbol{\Phi} (\widehat{\boldsymbol{\beta}})/\mathbb{E} \{\psi^{\prime}(\epsilon)\} - \mathbf{\Psi} (\widehat{\boldsymbol{\beta}}) ||_{_{\mathbf{G}}}
\\ & = \left\{  B( K_2 K_1^{-1} n^{-1} A_n^2)^{1/2} + B^{\prime} K_2 K_1^{-1/2} A_{n} C_n^{1/2} \right\} (K_1 C_n)^{1/2},
\end{align*}
where the inequality holds on an event of probability greater than $1-\delta$. Applying again the limit relations $ \lim n^{-1} A_n^2 = \lim A_{n}^2 C_n = 0$ shows that $||\widehat{f} - \widetilde{f} ||_n^2 = o_P(C_n)$. This concludes the proof of the theorem.

\end{proof}

\begin{proof}[Proof of Corollary 1]

The proof follows from Lemma 8 and Lemma 9 of \cite{Stone:1985} after identifying the density $w$ of A.2 with the density of $t_{i}$. By assumption $w$ fulfils Condition 1 of \cite{Stone:1985}, as it is bounded away from zero and infinity. Check also that $K = o(n)$ by A.3 and therefore there exists an $\alpha>0$ such that

\begin{equation*}
\lim_{n \to \infty} n^{\alpha-1} K = 0
\end{equation*}
in either of the two cases $K \asymp n^{ 1/(2p+1)}$ or $K \asymp n^{v}, \ v \geq 1/(2q+1)$. 
\end{proof}

\begin{proof}[Proof of Theorem 3]

We will show that under the conditions of Theorem 3, 

\begin{equation*}
\Pr\left[\text{there is a solution} \ \widehat{f}(\cdot) \ \text{to} \  \boldsymbol{\Phi}_{\widehat{\sigma}}(\boldsymbol{\beta}) = \mathbf{0} \ \text{satisfying} \ ||\widehat{f} - \widetilde{f}||_n^2 \leq \delta C_n   \right] \geq 1-\delta,
\end{equation*}
with $C_n := \mathbb{E}\{||\widetilde{f} - f||_n^2 \}$, where $\widetilde{f}$ denotes the zero of estimating equation \eqref{eq:20}. Since

\begin{equation*}
||\boldsymbol{\Phi}_{\widehat{\sigma}}(\boldsymbol{\beta}) - \boldsymbol{\Psi}_{\sigma}(\boldsymbol{\beta} ||_{\mathbf{G}} \leq ||\boldsymbol{\Phi}_{\widehat{\sigma}}(\boldsymbol{\beta}) - \boldsymbol{\Phi}_{\sigma}(\boldsymbol{\beta}) ||_{\mathbf{G}} + || \boldsymbol{\Phi}_{\sigma}(\boldsymbol{\beta}) -  \boldsymbol{\Psi}_{\sigma}(\boldsymbol{\beta}||_{\mathbf{G}},
\end{equation*}
and from the proof of Theorem 1 it may be deduced that on $F_n$ there exists a constant $D>0$ such that

\begin{equation*}
|| \boldsymbol{\Phi}_{\sigma}(\boldsymbol{\beta}) -  \boldsymbol{\Psi}_{\sigma}(\boldsymbol{\beta}||_{\mathbf{G}} \leq D C_n^{1/2},
\end{equation*}
the theorem will be proven with a fixed-point argument, provided that we can establish the existence of another constant $Z>0$ such that

\begin{equation}
\label{eq:29}
||\boldsymbol{\Phi}_{\widehat{\sigma}}(\boldsymbol{\beta}) - \boldsymbol{\Phi}_{\sigma}(\boldsymbol{\beta}) ||_{\mathbf{G}}\leq Z C_n^{1/2}.
\end{equation}
Hereafter, we shall use $Z > 0$ as a generic constant which may take different values at different appearance. To prove the theorem note that, as $\widehat{\sigma} \xrightarrow{P} \sigma$, for every $\delta>0$ we will have $\widehat{\sigma}^{-1} \in (\sigma^{-1}-\delta, \sigma^{-1} + \delta)$ with probability tending to one. This implies that $\widehat{\sigma}^{-1} \geq (2\sigma)^{-1}$ with probability tending to one. Choose $\epsilon = (2\sigma)^{-1}$ in B.5. Adding and subtracting $\psi( r_i(\boldsymbol{\beta})/\widehat{\sigma}) \sigma^{-1}$, the boundedness of $\psi$ implied by B.5 yields
\begin{align*}
\left|\psi\left( \frac{r_i(\boldsymbol{\beta}}{\widehat{\sigma}} \right) \frac{1}{\widehat{\sigma}} - \psi\left( \frac{r_i(\boldsymbol{\beta}}{\sigma} \right) \frac{1}{\sigma} \right| & \leq Z \frac{|\widehat{\sigma}-\sigma|}{\widehat{\sigma}\sigma} + M_{\sigma} \frac{|\widehat{\sigma}-\sigma|}{\widehat{\sigma}\sigma} \leq Z \frac{|\widehat{\sigma}-\sigma|}{2\sigma^2},
\end{align*}
for some $Z>0$ with probability tending to one. Thus, there exists a constant $Z>0$ depending only on $\sigma$ such that
\begin{equation*}
|||\boldsymbol{\Phi}_{\widehat{\sigma}}(\boldsymbol{\beta}) - \boldsymbol{\Phi}_{\sigma}(\boldsymbol{\beta}) ||_{\mathbf{G}} \leq Z A_n |\widehat{\sigma}-\sigma|  
\end{equation*}
with probability tending to one. But $n^{1/2}(\widehat{\sigma}-\sigma) = O_P(1)$, by B.4 and $A_n C_n^{-1/2} n^{-1/2} = O_P(1)$, by the definition of $A_n$, inequality  \eqref{eq:27} and Lemma \ref{Lem:2}(iii). This shows that there exists another constant $Z>0$ such that with high probability
\begin{equation*}
||\boldsymbol{\Phi}_{\widehat{\sigma}}(\boldsymbol{\beta}) - \boldsymbol{\Phi}_{\sigma}(\boldsymbol{\beta}) ||_{\mathbf{G}} \leq Z C_n^{1/2},
\end{equation*}
for all large $n$. The theorem is thus established.

\end{proof}

\end{document}